\documentclass[aps, pra, reprint, superscriptaddress]{revtex4-1}

\pdfoutput=1

\usepackage{header}

\graphicspath{{graphics/}}

\begin{document}

\bibliographystyle{apsrev4-1}

\title{Efficient Algorithms for Approximating Quantum Partition Functions}

\author{Ryan L. Mann}
\email{mail@ryanmann.org}
\homepage{http://www.ryanmann.org}
\affiliation{School of Mathematics, University of Bristol, Bristol, BS8 1UG, United Kingdom}

\author{Tyler Helmuth}
\email{tyler.helmuth@durham.ac.uk}
\homepage{http://www.tylerhelmuth.net}
\affiliation{School of Mathematics, University of Bristol, Bristol, BS8 1UG, United Kingdom}

\begin{abstract}
    We establish a polynomial-time approximation algorithm for partition functions of quantum spin models at high temperature. Our algorithm is based on the quantum cluster expansion of Neto\v{c}n\'y and Redig and the cluster expansion approach to designing algorithms due to Helmuth, Perkins, and Regts. Similar results have previously been obtained by related methods, and our main contribution is a simple and slightly sharper analysis for the case of pairwise interactions on bounded-degree graphs.
\end{abstract}

\maketitle

\section{Introduction}
\label{section:Introduction}

Classical algorithms for approximating partition functions of quantum models that make use of cluster expansions have occurred in two recent papers~\cite{harrow2020classical, kuwahara2020clustering}. In this paper we provide a simple and concise exposition of how to construct such algorithms, with the intent of making the technique accessible to a wide audience.

A \emph{quantum spin system} is modelled by a hypergraph \mbox{$G=(X,E)$}. At each vertex $x$ of $G$ there is a $d$-dimensional Hilbert space $\mathcal{H}_x$ with $d<\infty$. The Hilbert space on the hypergraph is given by \mbox{$\mathcal{H}_G\coloneqq\bigotimes_{x \in X}\mathcal{H}_x$}. An interaction $\Phi$ assigns a self-adjoint operator $\Phi(e)$ on \mbox{$\mathcal{H}_e\coloneqq\bigotimes_{x \in e}\mathcal{H}_x$} to each hyperedge $e$ of $G$. The Hamiltonian on $G$ is defined by \mbox{$H_G\coloneqq\sum_{e \in E(G)}\Phi(e)$}. We are interested in the \emph{quantum partition function} $Z_G(\beta)$ at inverse temperature $\beta$, defined by \mbox{$Z_G(\beta)\coloneqq\Tr\left[e^{-\beta H_G}\right]$}.

In what follows we shall focus our attention on quantum spin systems modelled by bounded-degree graphs, however generalisations to bounded-degree bounded-rank hypergraphs are also possible. We shall assume that \mbox{$\norm{\Phi(e)}\leq1$} for every $e \in E$, where $\norm{\;\cdot\;}$ denotes the operator norm. Note that this is always possible by a rescaling of $\beta$. To state our main result, recall that a \emph{fully polynomial-time approximation scheme} for a sequence of complex numbers $(z_n)_{n\in\mathbb{N}}$ is a deterministic algorithm that, for any $n$ and $\epsilon>0$, produces a complex number $\hat{z}_n$ such that \mbox{$\abs{z_n-\hat{z}_n}\leq\epsilon\abs{z_n}$} in time polynomial in $n$ and $1/\epsilon$.

\begin{theorem}
    \label{theorem:ApproximationAlgorithmPartitionFunction}
    Fix $\Delta\in\mathbb{Z}^+$. There is a fully polynomial-time approximation scheme for the partition function $Z_G(\beta)$ for all graphs $G$ of maximum degree at most $\Delta$ and all complex numbers $\beta$ such that \mbox{$\abs{\beta}\leq\frac{1}{e^4\Delta}$}.
\end{theorem}

Our algorithm is based on combining the abstract cluster expansion for quantum spin systems of Neto\v{c}n\'y and Redig~\cite{netocny2004large} with the algorithmic framework of Helmuth, Perkins, and Regts~\cite{helmuth2020algorithmic}. The condition \mbox{$\abs{\beta}=O\left(\frac{1}{\Delta}\right)$} is optimal under the complexity-theoretic assumption that \textsc{RP} (randomised polynomial time) is not equal to \textsc{NP} (non-deterministic polynomial time) due to results on the hardness of approximate counting~\cite{sly2012computational, galanis2016inapproximability}. We remark that these results concern real values of $\beta$; however, similar computational complexity transitions from \textsc{P} (polynomial time) to \mbox{\textsc{BQP}-hard} (bounded-error quantum polynomial time) and \textsc{P} to \mbox{\textsc{\#P}-hard} can also be observed for complex values of $\beta$ by using the methods of Refs.~\cite{bremner2010classical, goldberg2017complexity, mann2019approximation}. For a formal definition of these complexity-theoretic notions, we refer the reader to Ref.~\cite{arora2009computational}.

Previous work on polynomial-time approximate counting algorithms for classical models have typically followed one of three approaches: the correlation decay method~\cite{weitz2006counting, liu2019ising}, Markov-chain Monte Carlo~\cite{jerrum1993polynomial}, or interpolation-type methods~\cite{barvinok2016combinatorics, helmuth2020algorithmic}. The latter two of these methods have also been used to design classical algorithms for quantum models~\cite{bravyi2015monte, mann2019approximation, bravyi2019classical, harrow2020classical, kuwahara2020clustering, crosson2020classical}. The goal of this paper is to convey the simplicity and flexibility of the third method that results from using the cluster expansion formalism. We emphasise that ideas of this type have previously been used to establish similar algorithms: for \mbox{$\abs{\beta}\leq(10e^2\Delta)^{-1}$} with quasi-polynomial runtime~\cite{harrow2020classical}, and for \mbox{$\abs{\beta}\leq(16e^3\Delta)^{-1}$} with polynomial runtime~\cite{kuwahara2020clustering}. Both the runtime of our algorithm and that of Ref.~\cite{kuwahara2020clustering} are polynomials of a relatively high degree; examining our proof gives an upper bound of $O(\log(d\Delta))$ for the degree. While our results represent a modest improvement in the bound for $\abs{\beta}$, we view our main contribution as being the simplicity of our analysis.

We note also that \emph{a priori} information on the location of zeros of the partition function can be combined with the methods of this paper to develop polynomial-time algorithms. As noted in Ref.~\cite{helmuth2020algorithmic}, this is an alternate route to results of Patel and Regts~\cite{patel2017deterministic} using Barvinok's method~\cite{barvinok2016combinatorics}. For quasi-polynomial time results of this flavour in the quantum setting, see Ref.~\cite{harrow2020classical}.

This paper is structured as follows. In \mbox{Section~\ref{section:AbstractClusterExpansion}}, we introduce the abstract cluster expansion. Then, in \mbox{Section~\ref{section:QuantumClusterExpansion}}, we show how the partition function of quantum spins systems admits such a cluster expansion. In \mbox{Section~\ref{section:ApproximationAlgorithm}}, we use this framework to establish our approximation algorithm for the quantum partition function at high temperature. Finally, we conclude in \mbox{Section~\ref{section:ConclusionAndOutlook}} with some remarks and open problems.

\section{The Abstract Cluster Expansion}
\label{section:AbstractClusterExpansion}

The \emph{cluster expansion} is a powerful tool from mathematical physics that allows one to express, via power series expansions, perturbations of a well-understood reference model. When the perturbations are sufficiently small, the power series expansions are convergent and allow one to draw many conclusions regarding correlation decay, zero-freeness, and other related properties. This method was originally introduced by Mayer in the study of imperfect gases~\cite{mayer1940statistical}, but has since been greatly abstracted and simplified. The formulation we use is due to Koteck\'y and Preiss~\cite{kotecky1986cluster}. 

An \emph{abstract polymer model} is a triple $(\mathcal{C}, w, \sim)$, where $\mathcal{C}$ is a countable set whose objects are called \emph{polymers}, \mbox{$w:\mathcal{C}\to\mathbb{C}$} is a function that assigns to each polymer $\gamma\in\mathcal{C}$ a complex number $w_\gamma$ called the \emph{weight} of the polymer, and $\sim$ is a \emph{symmetric compatibility relation} such that each polymer is incompatible with itself. Equivalently, the \emph{incompatibility relation} $\nsim$ is a symmetric and reflexive relation. A set of polymers is called \emph{admissible} if all the polymers in the set are all pairwise compatible. Note that the empty set is admissible. Let $\mathcal{G}$ denote the collection of all admissible sets of polymers from $\mathcal{C}$. Then the abstract polymer partition function is defined by
\begin{equation}
    Z(\mathcal{C},w) \coloneqq \sum_{\Gamma\in\mathcal{G}}\prod_{\gamma\in\Gamma}w_\gamma. \notag
\end{equation}
Our algorithm is based on reformulating the partition function of a quantum spin system in the abstract polymer model language, see \mbox{Section~\ref{section:QuantumClusterExpansion}}. The utility of this is due to the following fact about $\log(Z(\mathcal{C},w))$.

Let $\Gamma$ be a non-empty ordered tuple of polymers. The \emph{incompatibility graph} $H_\Gamma$ of $\Gamma$ is the graph with vertex set $\Gamma$ and edges between any two polymers if and only if they are incompatible. $\Gamma$ is called a \emph{cluster} if its incompatibility graph $H_\Gamma$ is connected. Let $\mathcal{G}_C$ denote the set of all clusters of polymers from $\mathcal{C}$. The \emph{abstract cluster expansion}~\cite{kotecky1986cluster, friedli2017statistical} is a formal power series for $\log{Z(\mathcal{C},w)}$ in the variables $w_\gamma$, defined by
\begin{equation}
    \log(Z(\mathcal{C},w)) \coloneqq \sum_{\Gamma\in\mathcal{G}_C}\varphi(H_\Gamma)\prod_{\gamma\in\Gamma}w_\gamma, \notag
\end{equation}
where $\varphi(H)$ denotes the \emph{Ursell function} of a graph $H$:
\begin{equation}
    \varphi(H) \coloneqq \frac{1}{\abs{V(H)}!}\sum_{\substack{E \subseteq E(H) \\ \text{spanning} \\ \text{connected}}}(-1)^{\abs{E}}. \notag
\end{equation}

\section{The Quantum Cluster Expansion}
\label{section:QuantumClusterExpansion}

In this section we shall show how the partition function of a quantum spin system admits an abstract polymer representation and hence an abstract cluster expansion. We return to the more general setting of hypergraphs for the remainder of this section.

Consider a quantum spin system modelled by the hypergraph \mbox{$G=(X,E)$} with interaction $\Phi$, where at each vertex $x$ of $G$ there is a $d$-dimensional Hilbert space $\mathcal{H}_x$ with $d<\infty$. Recall that $\Phi$ assigns a self-adjoint operator $\Phi(e)$ on \mbox{$\mathcal{H}_e\coloneqq\bigotimes_{x \in e}\mathcal{H}_x$} to each hyperedge $e$ of $G$. Define a polymer $\gamma$ in this model to be a multiset \mbox{$(E_\gamma, m_\gamma)$} of hyperedges $E_\gamma \subseteq E$ with multiplicity function \mbox{$m_\gamma\colon E_\gamma\to\mathbb{Z}^+$} whose support $E_\gamma$ induces a connected subgraph. Say that two polymers are compatible if and only if their supporting subgraphs are vertex disjoint. For a polymer $\gamma$ let \mbox{$\norm{\gamma}\coloneqq\sum_{e \in E_\gamma}m_\gamma(e)$} denote its size and let \mbox{$\abs{\gamma}\coloneqq\abs{E_\gamma}$} denote the cardinality of its support; by a slight abuse of notation we will write $\gamma=\{\gamma_i\}_{i=1}^{\norm{\gamma}}$. With these definitions, the partition function $Z_G(\beta)$ admits an abstract polymer model representation~\cite{park1982cluster, ueltschi1998discontinuous, netocny2004large} as formalised by the following lemma.
\begin{lemma}[{restate=[name=restatement]QuantumAbstractPolymerModel}]
    \label{lemma:QuantumAbstractPolymerModel}
    The partition function $Z_G(\beta)$ admits the following abstract polymer model representation.
    \begin{equation}
        Z_G(\beta) = \sum_{\Gamma\in\mathcal{G}}\prod_{\gamma\in\Gamma}w_\gamma, \notag
    \end{equation}
    where
    \begin{equation}
        w_\gamma \coloneqq \frac{(-\beta)^{\norm{\gamma}}}{\norm{\gamma}!\prod_{e \in E_\gamma}m_\gamma(e)!}\Tr\left[\sum_{\sigma \in S_{\norm{\gamma}}}\prod_{i=1}^{\norm{\gamma}}\Phi(\gamma_{\sigma(i)})\right]. \notag
    \end{equation}
\end{lemma}
Note that $S_{\norm{\gamma}}$ is the symmetric group of degree $\norm{\gamma}$ and this assigns an ordering to the product of interactions. We prove \mbox{Lemma~\ref{lemma:QuantumAbstractPolymerModel}} in \mbox{Appendix~\ref{section:QuantumAbstractPolymerModel}}. Note that the abstract polymer model representation holds as a formal power series in $\beta$. As an immediate corollary, we obtain a cluster expansion for $\log(Z_G(\beta))$.

\begin{corollary}
    The partition function $Z_G(\beta)$ admits the following cluster expansion.
    \begin{equation}
        \log(Z_G(\beta)) \coloneqq \sum_{\Gamma\in\mathcal{G}_C}\varphi(H_\Gamma)\prod_{\gamma\in\Gamma}w_\gamma. \notag
    \end{equation}
\end{corollary}

For algorithms, an important quantity is the \emph{truncated cluster expansion} for $\log(Z_G(\beta))$:
\begin{equation}
    T_m(Z_G(\beta)) \coloneqq \sum_{\substack{\Gamma\in\mathcal{G}_C \\ \norm{\Gamma} < m}}\varphi(H_\Gamma)\prod_{\gamma\in\Gamma}w_\gamma, \notag
\end{equation}
where \mbox{$\norm{\Gamma}\coloneqq\sum_{\gamma\in\Gamma}\norm{\gamma}$}.

\section{Approximation Algorithm}
\label{section:ApproximationAlgorithm}

We now establish our approximation algorithm. Firstly, we show that truncated cluster expansion provides a good approximation to $\log(Z_G(\beta))$. Neto\v{c}n\'y and Redig~\cite{netocny2004large} provided a sufficient condition for the convergence of the quantum cluster expansion based on the formalism of Koteck\'y and Preiss~\cite{kotecky1986cluster}. In the following lemma, we follow their analysis in the setting of bounded-degree graphs. In particular, we obtain convergence criteria based on the maximum degree alone.
\begin{lemma}[{restate=[name=restatement]TruncatedQuantumClusterExpansionError}]
    \label{lemma:TruncatedQuantumClusterExpansionError}
    Fix $\Delta\in\mathbb{Z}^+$. Let \mbox{$G=(V,E)$} be a graph of maximum degree at most $\Delta$ and let $\beta$ be a complex number such that \mbox{$\abs{\beta}\leq\frac{1}{e^4\Delta}$}. The cluster expansion for $\log(Z_G(\beta))$ converges absolutely, $Z_G(\beta)\neq0$, and for $m\in\mathbb{Z}^+$,
    \begin{equation}
        \abs{T_m(Z_G(\beta))-\log(Z_G(\beta))} \leq \abs{V}e^{-m}. \notag
    \end{equation}
\end{lemma}
We prove \mbox{Lemma~\ref{lemma:TruncatedQuantumClusterExpansionError}} in \mbox{Appendix~\ref{section:TruncatedQuantumClusterExpansionErrorAndPolyerWeightsAlgorithm}}. This lemma implies that to obtain an multiplicative $\epsilon$-approximation to $Z_G(\beta)$, it is sufficient to compute $T_m(Z_G(\beta))$ to order \mbox{$m=\log\left(\abs{V(G)}/\epsilon\right)$}. We shall proceed by establishing an algorithm for computing $T_m(Z_G(\beta))$ in time \mbox{$\exp(O(m))\cdot\abs{V(G)}^{O(1)}$}. Helmuth, Perkins, and Regts~\cite{helmuth2020algorithmic} showed that such an algorithm exists given the following three lemmas.
\begin{lemma}
    \label{lemma:ListClustersAlgorithm}
    Fix $\Delta\in\mathbb{Z}^+$, and let \mbox{$G=(V,E)$} be a graph of maximum degree at most $\Delta$. The clusters of size at most $m$ can be listed in time \mbox{$\exp(O(m))\cdot\abs{V}^{O(1)}$}.
\end{lemma}
\begin{proof}
    Our proof follows that of Ref.~\cite[Theorem 6]{helmuth2020algorithmic}. Firstly, we enumerate all connected subgraphs in $G$ of size at most $m$. This can be achieved in time \mbox{$\exp(O(m))\cdot\abs{V}^{O(1)}$} by Ref.~\cite[Lemma 3.6]{patel2017deterministic}. Then, for each subgraph $H$, we enumerate all polymers (multisets) of size at most $m$ whose corresponding subgraph in $G$ is $H$. If $H$ has size $n$, then there are precisely $\binom{m-1}{n-1}$ of these and they can be enumerated in time $\exp(O(m))$. The enumeration of clusters in the claimed time then follows as in the proof of Ref.~\cite[Theorem 6]{helmuth2020algorithmic}.
\end{proof}

\begin{lemma}
    \label{lemma:UrsellFunctionAlgorithm}
    The Ursell function $\varphi(H)$ can be computed in time $\exp(O(\abs{V(H)}))$.
\end{lemma}
\begin{proof}
    This is a result of Ref.~\cite{bjorklund2008computing}; see Ref.~\cite[Lemma~5]{helmuth2020algorithmic}.
\end{proof}

\begin{lemma}[{restate=[name=restatement]PolyerWeightsAlgorithm}]
    \label{lemma:PolyerWeightsAlgorithm}
    The weight $w_\gamma$ of a polymer $\gamma$ can be computed in time $\exp(O(\norm{\gamma}))$.
\end{lemma}
We prove \mbox{Lemma~\ref{lemma:PolyerWeightsAlgorithm}} in \mbox{Appendix~\ref{section:TruncatedQuantumClusterExpansionErrorAndPolyerWeightsAlgorithm}}.

\begin{lemma}
    \label{lemma:TruncatedQuantumClusterAlgorithm}
    Fix $\Delta\in\mathbb{Z}^+$, and let \mbox{$G=(V,E)$} be a graph of maximum degree at most $\Delta$. The truncated cluster expansion $T_m(Z_G(\beta))$ can be computed in time \mbox{$\exp(O(m))\cdot\abs{V(G)}^{O(1)}$}.
\end{lemma}
\begin{proof}
    We can list all clusters in $G$ of size at most $m$ in time \mbox{$\exp(O(m))\cdot\abs{V}^{O(1)}$} by \mbox{Lemma~\ref{lemma:ListClustersAlgorithm}}. For each of these clusters, we can compute the Ursell function in time $\exp(O(m))$ by \mbox{Lemma~\ref{lemma:UrsellFunctionAlgorithm}}, and the polymer weights in time $\exp(O(m))$ by \mbox{Lemma~\ref{lemma:PolyerWeightsAlgorithm}}. Hence, the truncated cluster expansion for $\log(Z_G(\beta))$ can be computed in time \mbox{$\exp(O(m))\cdot\abs{V(G)}^{O(1)}$}. 
\end{proof}

Combining \mbox{Lemma~\ref{lemma:TruncatedQuantumClusterExpansionError}} and \mbox{Lemma~\ref{lemma:TruncatedQuantumClusterAlgorithm}} gives a fully polynomial-time approximation scheme for the partition function $Z_G(\beta)$ when $G$ has maximum degree at most $\Delta$ and $\abs{\beta}$ is at most $\frac{1}{e^4\Delta}$. This proves \mbox{Theorem~\ref{theorem:ApproximationAlgorithmPartitionFunction}}. \linebreak

\section{Conclusion \& Outlook}
\label{section:ConclusionAndOutlook}

We have discussed how classical algorithms based on cluster expansion methods apply to quantum spin systems at high temperature. Our focus has been on conveying the simplicity of the method, which has appeared previously in other forms~\cite{harrow2020classical, kuwahara2020clustering}. We note that it may be possible to use the Markov chain polymer approach of Ref.~\cite{chen2019fast} to obtain an algorithm with an improved runtime.

For discrete classical spin systems, expansion methods have also been used at low temperatures, i.e., when $\beta\gg1$~\cite{helmuth2020algorithmic, jenssen2019algorithms, liao2019counting, borgs2020efficient, carlson2020efficient}. It would be interesting to adapt these methods to quantum systems, e.g., by developing algorithms based on Pirogov-Sinai methods for quantum perturbations of classical systems~\cite{borgs1996low}. We remark, however, that it seems difficult to use this approach for low-temperature quantum systems with an infinite degeneracy of ground states, for example, when the set of ground states possesses a continuous symmetry.

\section*{Acknowledgements}

We thank Michael Bremner, Adrian Chapman, Ashley Montanaro, and Will Perkins for helpful discussions. RLM was supported by the QuantERA ERA-NET Cofund in Quantum Technologies implemented within the European Union's Horizon 2020 Programme (QuantAlgo project) and EPSRC grants EP/L021005/1 and EP/R043957/1. TH was supported by EPSRC grant EP/P003656/1. No new data were created during this study.

\onecolumngrid

\appendix

\section{Proof of Lemma~\ref*{lemma:QuantumAbstractPolymerModel}}
\label{section:QuantumAbstractPolymerModel}

\QuantumAbstractPolymerModel*

\begin{proof}
    Our proof follows Neto\v{c}n\'y and Redig~\cite{netocny2004large}. By definition,
    \begin{align}
        Z_G(\beta) &= \Tr\left[e^{-\beta H_G}\right] \notag \\
        &= \sum_{n=0}^\infty\frac{(-\beta)^n}{n!}\Tr\left[\left(\sum_{e \in E(G)}\Phi(e)\right)^n\right] \notag \\
        &= \sum_{n=0}^\infty\frac{(-\beta)^n}{n!}\sum_{e_1,\ldots,e_n \in E(G)}\Tr\left[\prod_{i=1}^n\Phi(e_i)\right]. \notag
    \end{align}
    We shall now rewrite the inner sum as a product over disjoint objects. Let $S=(e_i)_{i=1}^n$ denote any sequence of hyperedges and let $G_S$ be the graph with vertex set $[n]$ and edges between any two vertices $i$ and $j$ if and only if $e_i \cap e_j \neq \varnothing$. Define a sequential polymer to be a subsequence of $S$ that corresponds to a maximally connected component of $G_S$. Say that two sequential polymers are compatible if and only if their corresponding subgraphs in $G$ are vertex disjoint. Let $\Gamma_S$ denote the set of all sequential polymers in $S$. It follows that,
    \begin{equation}
        Z_G(\beta) = \sum_{n=0}^\infty\frac{(-\beta)^n}{n!}\sum_{S:\abs{S}=n}\prod_{\gamma\in\Gamma_S}\Tr\left[\prod_{e\in\gamma}\Phi(e)\right]. \notag
    \end{equation}
    Let $\abs{\gamma}$ denote the length of a sequential polymer. Further, let $\Gamma_G\coloneqq\cup_S\Gamma_S$ denote the set of all sequential polymers in $G$ and let $\mathcal{G}_G$ denote the collection of all admissible sets of sequential polymers in $G$. Observe that for any admissible set of sequential polymers $\{\gamma_i\}_{i=1}^k$ there are precisely $\frac{\left(\sum_{i=1}^k\abs{\gamma_i}\right)!}{\prod_{i=1}^k\abs{\gamma_i}!}$ sequences $S$ that give rise to it. Thus, we may write
    \begin{align}
        Z_G(\beta) &= \sum_{n=0}^\infty\frac{(-\beta)^n}{n!}\sum_{k=0}^n\frac{1}{k!}\sum_{\substack{m_1,\ldots,m_k\geq1 \\ m_1+\ldots+m_k=n}}\binom{n}{m_1,\ldots,m_k}\sum_{\substack{\gamma_1,\ldots,\gamma_k\in\Gamma_G \\ \abs{\gamma_1}=m_1,\ldots,\abs{\gamma_k}=m_k \\ \text{admissible}}}\prod_{i=1}^k\Tr\left[\prod_{e\in\gamma}\Phi(e)\right] \notag \\
        &= \sum_{n=0}^\infty\sum_{k=0}^n\frac{1}{k!}\sum_{\substack{m_1,\ldots,m_k\geq1 \\ m_1+\ldots+m_k=n}}\sum_{\substack{\gamma_1,\ldots,\gamma_k\in\Gamma_G \\ \abs{\gamma_1}=m_1,\ldots,\abs{\gamma_k}=m_k \\ \text{admissible}}}\prod_{i=1}^k\frac{(-\beta)^{\abs{\gamma_i}}}{\abs{\gamma_i}!}\Tr\left[\prod_{e\in\gamma}\Phi(e)\right]. \notag
    \end{align}
    By interchanging the summations over $n$ and $k$, we obtain
    \begin{align}
        Z_G(\beta) &= \sum_{k=0}^\infty\frac{1}{k!}\sum_{n=k}^\infty\sum_{\substack{m_1,\ldots,m_k\geq1 \\ m_1+\ldots+m_k=n}}\sum_{\substack{\gamma_1,\ldots,\gamma_k\in\Gamma_G \\ \abs{\gamma_1}=m_1,\ldots,\abs{\gamma_k}=m_k \\ \text{admissible}}}\prod_{i=1}^k\frac{(-\beta)^{\abs{\gamma_i}}}{\abs{\gamma_i}!}\Tr\left[\prod_{e\in\gamma}\Phi(e)\right] \notag \\
        &= \sum_{k=0}^\infty\frac{1}{k!}\sum_{\substack{\gamma_1,\ldots,\gamma_k\in\Gamma_G \\ \text{admissible}}}\prod_{i=1}^k\frac{(-\beta)^{\abs{\gamma_i}}}{\abs{\gamma_i}!}\Tr\left[\prod_{e\in\gamma}\Phi(e)\right] \notag \\
        &= \sum_{\Gamma\in\mathcal{G}_G}\prod_{\gamma\in\Gamma}\frac{(-\beta)^{\abs{\gamma}}}{\abs{\gamma}!}\Tr\left[\prod_{e\in\gamma}\Phi(e)\right]. \notag
    \end{align}
    Now, by transforming the sum over admissible sets of sequential polymers into a sum over admissible sets of polymers and summing the weights of their permutations, we obtain
    \begin{equation}
        Z_G(\beta) = \sum_{\Gamma\in\mathcal{G}}\prod_{\gamma\in\Gamma}\frac{(-\beta)^{\norm{\gamma}}}{\norm{\gamma}!\prod_{e \in E_\gamma}m_\gamma(e)!}\Tr\left[\sum_{\sigma \in S_{\norm{\gamma}}}\prod_{i=1}^{\norm{\gamma}}\Phi(\gamma_{\sigma(i)})\right]. \notag
    \end{equation}
    Since there are equivalent sequential polymers being distinguished in the sum over permutations, i.e., permutations of repeated hyperedges, we had to introduce a factor of $\frac{1}{\prod_{e \in E_\gamma}m_\gamma(e)!}$ to avoid overcounting. This completes the proof.
\end{proof}

\section{Proof of Lemma~\ref*{lemma:TruncatedQuantumClusterExpansionError} and Lemma~\ref*{lemma:PolyerWeightsAlgorithm}}
\label{section:TruncatedQuantumClusterExpansionErrorAndPolyerWeightsAlgorithm}

\TruncatedQuantumClusterExpansionError*

\begin{proof}
    For convenience, we normalise the trace so that $\Tr(\mathbb{I})=1$. Note that this is equivalent to a rescaling of the partition function by a multiplicative factor. We introduce a polymer $\gamma_x$ to every vertex $x$ in $G$ consisting of only that vertex. We define $\gamma_x$ to be incompatible with every polymer that contains $x$. Then, we have
    \begin{equation}
        \sum_{\gamma\nsim\gamma_x}\abs{w_\gamma}e^{\abs{\gamma}+\norm{\gamma}+1} \leq \sum_{\gamma\nsim\gamma_x}\frac{e^{\abs{\gamma}+\norm{\gamma}+1}\abs{\beta}^{\norm{\gamma}}}{\prod_{e \in E_\gamma}m_\gamma(e)!}\prod_{i=1}^{\norm{\gamma}}\norm{\Phi(\gamma_i)} \leq \sum_{\gamma\nsim\gamma_x}\frac{e^{\abs{\gamma}+\norm{\gamma}+1}\abs{\beta}^{\norm{\gamma}}}{\prod_{e \in E_\gamma}m_\gamma(e)!} \leq e\sum_{\gamma\nsim\gamma_x}(2e)^{\abs{\gamma}}\left(\frac{e\abs{\beta}}{2}\right)^{\norm{\gamma}}. \notag
    \end{equation}
    For a vertex $x$, there are at most $\frac{(e\Delta)^n}{2}$ connected subgraphs with $n$ edges that contain $x$~\cite[Lemma 2.1]{borgs2013left}. Furthermore, for such a subgraph, there are precisely $\binom{k-1}{n-1}$ polymers (multisets) $\gamma$ with $\norm{\gamma}=k$ that correspond to it. Thus, we may write
    \begin{equation}
        \sum_{\gamma\nsim\gamma_x}\abs{w_\gamma}e^{\abs{\gamma}+\norm{\gamma}+1} \leq \frac{e}{2}\sum_{n=1}^\infty(2e^2\Delta)^n\sum_{k=n}^\infty\binom{k-1}{n-1}\left(\frac{e\abs{\beta}}{2}\right)^k. \notag
    \end{equation}
    By interchanging the summations over $n$ and $k$, we obtain
    \begin{align}
        \sum_{\gamma\nsim\gamma_x}\abs{w_\gamma}e^{\abs{\gamma}+\norm{\gamma}+1} &\leq \frac{e}{2}\sum_{k=1}^\infty\left(\frac{e\abs{\beta}}{2}\right)^k\sum_{n=1}^k\binom{k-1}{n-1}(2e^2\Delta)^n \notag \\
        &= \frac{e}{2}\sum_{k=1}^\infty\left(\frac{e\abs{\beta}}{2}\right)^k(2e^2\Delta)(2e^2\Delta+1)^{k-1} \notag \\
        &\leq \frac{e}{2}\sum_{k=1}^\infty\abs{\beta}^k\left(e^3\Delta+\frac{e}{2}\right)^k. \notag
    \end{align}
    By taking \mbox{$\abs{\beta}\leq\frac{1}{e^4\Delta}$}, we have
    \begin{equation}
        \sum_{\gamma\nsim\gamma_s}\abs{w_\gamma}e^{\abs{\gamma}+\norm{\gamma}+1} \leq \frac{e}{2}\sum_{k=1}^\infty \left(\frac{1}{e}+\frac{1}{e^3}\right)^k < 1. \notag
    \end{equation}
    Fix a polymer $\gamma$. By summing over all vertices in $\gamma$, of which there are at most $\abs{\gamma}+1$, we obtain
    \begin{equation}
        \sum_{\gamma^*\nsim\gamma}\abs{w_{\gamma^*}}e^{\abs{\gamma^*}+\norm{\gamma^*}+1} \leq \abs{\gamma}+1. \notag
    \end{equation}
    Now, by applying the main theorem of Ref.~\cite{kotecky1986cluster} with \mbox{$a(\gamma)=\abs{\gamma}+1$} and \mbox{$d(\gamma)=\norm{\gamma}$}, we obtain that the cluster expansion converges absolutely in $\beta$, $Z_G(\beta)\neq0$, and
    \begin{equation}
        \sum_{\substack{\Gamma\in\mathcal{G}_C \\ \Gamma \ni \gamma_x}}\abs{\varphi(\Gamma)\prod_{\gamma\in\Gamma}w_\gamma}e^{\norm{\Gamma}} \leq 1. \notag
    \end{equation}
    We complete the proof by summing over all vertices in $G$:
    \begin{equation*}
        \sum_{\substack{\Gamma\in\mathcal{G}_C \\ \norm{\Gamma} \geq m}}\abs{\varphi(\Gamma)\prod_{\gamma\in\Gamma}w_\gamma} \leq \abs{V}e^{-m}. \qedhere
    \end{equation*}
\end{proof}

\PolyerWeightsAlgorithm*

\begin{proof}
    For convenience, let $n=\norm{\gamma}$ and fix an enumeration $\gamma_{1},\dots, \gamma_n$ of the multiset of edges in $\gamma$. By an inclusion-exclusion argument as in the derivation of Ryser's formula for the permanent~\cite{ryser1963combinatorial}, we have
    \begin{equation}
        \sum_{\sigma \in S_n}\prod_{i=1}^n\Phi(\gamma_{\sigma(i)}) = (-1)^n\sum_{A\subseteq[n]}(-1)^{\abs{A}}\left(\sum_{i \in A}\Phi(\gamma_i)\right)^n. \notag
    \end{equation}
    Thus, we may write
    \begin{equation}
        w_\gamma = \frac{(-\beta)^n}{n!\prod_{e \in E_\gamma}m_\gamma(e)!}\Tr\left[\sum_{\sigma \in S_n}\prod_{i=1}^n\Phi(\gamma_{\sigma(i)})\right] =  \frac{\beta^n}{n!\prod_{e \in E_\gamma}m_\gamma(e)!}\sum_{A\subseteq[n]}(-1)^{\abs{A}}\Tr\left[\left(\sum_{i \in A}\Phi(\gamma_i)\right)^n\right]. \notag
    \end{equation}
    The first sum is over all subsets of $[n]$, of which there are $2^n$. For each of these subsets $A$, we diagonalise the sum of the interactions $\sum_{i \in A}\Phi(\gamma_i)$ to obtain the eigenvalues. This can be achieved in time $\exp(O(n))$; here we are using our assumption that the single-spin Hilbert spaces $\mathcal{H}_{x}$ are $d$-dimensional with $d<\infty$. The trace may then be evaluated in time $\exp(O(n))$ by evaluating the sum of the $n^\mathrm{th}$ powers of the eigenvalues. Hence, the weight of a polymer can be computed in time \mbox{$\exp(O(n))=\exp(O(\norm{\gamma}))$}.
\end{proof}

\twocolumngrid

\bibliography{bibliography}

\end{document}